\documentclass[pra,twocolumn,showpacs]{revtex4-2}

\usepackage{amsmath}
\usepackage{latexsym}
\usepackage{amssymb}
\usepackage{graphicx}
\usepackage[colorlinks=true, citecolor=blue, urlcolor=blue]{hyperref}
\usepackage{float}
\usepackage{amsfonts}
\usepackage{textcomp}
\usepackage{mathpazo}
\usepackage{comment}
\usepackage{xr}

\usepackage{bbm}

\usepackage{xcolor}
\definecolor{myurlcolor}{rgb}{0,.5,.5}
\definecolor{mycitecolor}{rgb}{0,.6,0}
\definecolor{myrefcolor}{rgb}{2,0,0}
\usepackage{hyperref}
\hypersetup{colorlinks,
linkcolor=myrefcolor,
citecolor=mycitecolor,
urlcolor=myurlcolor}

\usepackage[draft]{fixme}
\usepackage{amsmath,bbm}
\usepackage{graphicx}
\usepackage{amsfonts}
\usepackage{amssymb}
\usepackage{amsmath, amssymb, amsthm,verbatim,graphicx,bbm}
\usepackage{mathrsfs}
\usepackage{color,xcolor,longtable}

\usepackage{xr}
\makeatletter
\newcommand*{\addFileDependency}[1]{
  \typeout{(#1)}
  \@addtofilelist{#1}
  \IfFileExists{#1}{}{\typeout{No file #1.}}
}
\makeatother

\newcommand*{\myexternaldocument}[1]{
    \externaldocument{#1}
    \addFileDependency{#1.tex}
    \addFileDependency{#1.aux}
}

\myexternaldocument{manuscript_PRL}

\newcommand{\beq}[0]{\begin{equation}}
\newcommand{\eeq}[0]{\end{equation}}

\newcommand{\one}{\leavevmode\hbox{\small1\normalsize\kern-.33em1}}

\def\be{\begin{equation}}
\def\ee{\end{equation}}
\def\ben{\begin{eqnarray}}
\def\een{\end{eqnarray}}
\def\eea{\end{array}}
\def\bea{\begin{array}}

\newcommand{\Tr}[1]{\mathrm{Tr}#1}
\newcommand{\bei}{\begin{itemize}}
\newcommand{\eei}{\end{itemize}}
\newcommand{\ket}[1]{|#1\rangle}
\newcommand{\bra}[1]{\langle#1|}

\newcommand{\proj}[1]{\ket{#1}\!\!\bra{#1}}

\newcommand{\I}{\mathbbm{1}}

\renewcommand{\emph}[1]{\textbf{#1}}

\makeatletter
\newtheorem*{rep@theorem}{\rep@title}
\newcommand{\newreptheorem}[2]{%
\newenvironment{rep#1}[1]{%
 \def\rep@title{#2 \ref{##1}}%
 \begin{rep@theorem}}%
 {\end{rep@theorem}}}
\makeatother

\theoremstyle{plain}
\newtheorem{thm}{Theorem}
\newtheorem*{thm*}{Theorem}
\newreptheorem{thm}{Theorem}

\newtheorem{fakt}{Fact}

\newtheorem{defn}[thm]{Definition}

\theoremstyle{definition}

\theoremstyle{remark}

\usepackage[T1]{fontenc}

\begin{document}

\title{Detecting entanglement in any measurement using quantum networks}
\author{Shubhayan Sarkar}
\email{shubhayan.sarkar@ulb.be}
\affiliation{Laboratoire d’Information Quantique, Université libre de Bruxelles (ULB), Av. F. D. Roosevelt 50, 1050 Bruxelles, Belgium}

\begin{abstract}	
Entanglement is a key resource to demonstrate quantum advantage over classical strategies. Entanglement in quantum states is one of the most well-explored areas in quantum physics. However, a rigorous approach to understanding and detecting entanglement in composite quantum measurements is lacking. In this work, we focus on composite quantum measurements and classify them into two classes: entangled and separable measurements. As done for quantum states, we define analogously a notion of witness that can be used to detect entanglement in composite quantum measurements. Here, one does not need to trust the measurement to witness its entanglement but must trust the quantum states. We then further extend this approach to show that any entangled measurement provides an advantage in network quantum steering without inputs, also known as swap steering. Consequently, this provides a way to witness entanglement in any quantum measurement in a one-sided device-independent way. Finally, we consider the star network scenario and show that any rank-one projective entangled quantum measurement gives a quantum advantage. Thus, one can detect the entanglement in any rank-one projective measurement in a device-independent way.
\end{abstract}


\maketitle

{\it{Introduction---}} Entanglement is a fundamental resource in quantum mechanics that enables the demonstration of a quantum advantage over classical strategies. Entanglement in quantum states has been extensively studied in the past decades \cite{Nielsen_Chuang_2010, QUANTUMENT, Guhne_2009}. On the other hand, entanglement in composite or joint quantum measurements have not been well-explored even in the simplest scenario when they act on two subsystems \cite{isoent}. Despite their importance in several quantum information protocols such as
entanglement swapping \cite{swap}, teleportation \cite{TELEPORTATION}, and dense coding
\cite{DENSECODING}, there is neither a systematic characterization of such measurements nor a useful way to detect them in experiments. As a matter of fact, there are projective entangled measurements that possess highly unintuitive features such as the measurement elements despite being orthogonal cannot be perfectly distinguished using only local operations and classical communication, also known as unextentible product bases \cite{UPB}. Moreover, there is an indication that entangled measurements are crucial if one wants to identify correlations that separate quantum theory from other generalised probabilistic theories \cite{GPT1}. Consequently, understanding entangled measurements has been posed as one of the important problems in quantum foundations that deserve extensive exploration
\cite{Cavalcanti2023}.

Entangled measurements are also desirable from a practical perspective. The vision of quantum internet is incomplete without entangled measurements. For instance, to establish quantum correlations over long distances, intermediate nodes equipped with quantum repeaters will be necessary \cite{repeater}. These repeaters perform joint measurements on the received systems to facilitate entanglement distribution. Also, the recently contrived ideas in quantum networks that do not require inputs to witness nonclassicality require entangled measurements, for instance, Refs. \cite{renou1, renou3, Tavakoli_2014, netstee,Sarkar2024networkquantum} to name a few. Moreover, they are also useful in self-testing any pure entangled state \cite{Supic2023} and any measurement along with mixed states \cite{sarkar2024universal}. They have also been implemented in numerous experiments, for instance in Refs. \cite{entexp1, entexp2, entexp3, entexp4, entexp5} to name a few. 

In this work, we delve into the study of composite quantum measurements and propose the simplest classification that divides them into two distinct categories: entangled and separable measurements. For a note, a classification of entangled measurements based on the cost of localisation restricted to the bipartite case has been previously considered in \cite{pauwels2024}. Our classification mirrors the approach used for quantum states, where we introduce a concept analogous to the notion of entanglement witnesses. These witnesses serve as tools to detect the presence of entanglement within any composite quantum measurements. Importantly, in this context, detecting entanglement does not require trusting the measurement process itself but instead relies on trusting the quantum states involved. 

Building on this framework, we extend our approach to demonstrate a significant application: any entangled measurement provides an advantage in network quantum steering without inputs, also referred to as swap steering \cite{Sarkar2024networkquantum}. This observation offers a novel perspective by establishing a method to witness the entanglement in any quantum measurement in a one-sided device-independent manner. This means that while one side of the system requires trust, the other side on which the entangled measurement acts remains untrusted. Finally, we consider the star network scenario and show that any rank-one projective entangled quantum measurement exhibits a quantum advantage. This further enables one to detect entanglement in rank-one projective composite measurements in a device-independent manner, where none of the experiment's devices need to be trusted. Moreover, in the case of measurements acting on two systems, the constructed witness is universal, that is, every rank-one bipartite projective entangled measurement can be detected using a single witness. 

{\it{Quantum measurement tomography---}} We begin by describing the tomography of quantum measurements. 
Consider a quantum measurement $\{M_i\}$ where $M_i$ denotes the measurement elements such that they are positive semi-definite and $\sum_iM_i=\I$. For projective measurements, an additional condition is $M_iM_j=M_i\delta_{ij}$. Considering a set of informational complete density matrices $\rho_j$ that are known (trusted), one can reconstruct the measurement by evaluating the probabilities, also referred to as correlations, $p_{ij}=\Tr(M_i\rho_j)$. 
Here one additionally assumes the dimension of the Hilbert space on which the measurements act. A broader class of tomography is known as process tomography where one considers tomography of quantum channels and has to trust the input states along with measurements.

At most times, it will not be required to reconstruct the complete measurement but only to confirm some relevant properties about it. 
Here we are interested in the entanglement in composite quantum measurements and thus classify them into two classes:

\begin{defn}[Separable measurements] \label{defsep} Consider the measurement $M_i$ such that all the measurement elements $\{M_i\}$ are separable, that is, $M_i\propto \sum_j\sigma_{i,j}\otimes\sigma'_{i,j}\otimes\ldots\ \forall i$, then  $\{M_i\}$ is a separable measurement. 
\end{defn}
Using the above definition, we define entangled measurements.
\begin{defn}[Entangled measurements] \label{defent} Consider the measurement $\{E_i\}$ such that at least one of the  measurement elements $E_i$ is not separable, then  $\{E_i\}$ is an entangled measurement. 
\end{defn}

Let us now construct witnesses that can detect whether a given composite measurement is separable or entangled.  

{\it{Witnesses---}} It is well-known that due to Hahn-Banach theorem, for every entangled state $\rho_e$ there exists an entanglement witness $\mathcal{W}_{\rho_e}$ such that $\Tr(\mathcal{W}_{\rho_e}\rho_e)<0$ and $\Tr(\mathcal{W}_{\rho_e}\sigma)\geq0$ where $\sigma$ represents any separable state. Similarly, we can find a general result to witness entanglement in quantum measurements that straightaway follows from the Hahn-Banach theorem.
\begin{fakt} For any entangled measurement $\{E_i\}$ def. \ref{defent}, there exists an entanglement witness $\mathcal{W}$ such that 
\begin{equation}\label{eq1}
   \min_i\Tr(\mathcal{W}E_i)<0,\quad \min_i\Tr(\mathcal{W}M_i)\geq0
\end{equation}
for every separable measurement $\{M_i\}$.
\end{fakt}
\begin{proof}
    Consider that the element $E_k$ of $\{E_i\}$ is entangled. Considering the witness $\mathcal{W}_{E_k}$, we have that $\Tr(\mathcal{W}_{E_k}E_k)<0$ and $\Tr(\mathcal{W}_{E_k}\sigma)\geq0$ for every separable state $\sigma$. As witnesses exist for every entangled state, thus every entangled measurement can be witnessed using the construction \eqref{eq1}.
\end{proof}

\begin{figure}[t]
\begin{center}
\includegraphics[width=.8\linewidth]{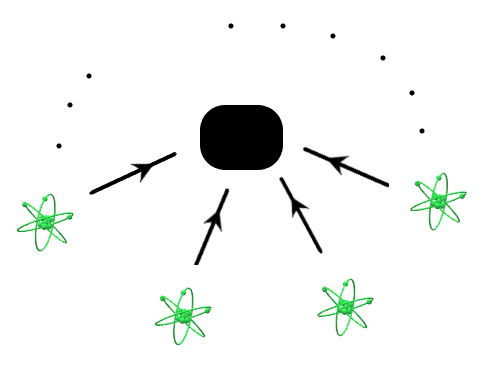}
    \caption{Witnessing entanglement in measurements. To detect entanglement in the measurement acting on $N$ subsystems, one needs to consider $N$ independent trusted sources that generate local states and send them to the untrusted measurement. From the obtained statistics, one can infer whether the measurement has entanglement or not.}
    \label{fig1}
    \end{center}
\end{figure}
We do not delve into the details of further construction of the witnesses of entanglement in measurements as it follows directly from constructing witnesses of entangled states which has been extensively studied \cite{QUANTUMENT, Guhne_2009}. We focus here on the experimental implementation of these witnesses. As an example, let us consider the Bell-basis measurement (BM) $\{\proj{\phi_{+}},\proj{\phi_{-}},\proj{\psi_{+}},\proj{\psi_{-}}\}$ where 
\begin{equation}\label{Amea1}
    \ket{\phi_{\pm}}=\frac{1}{\sqrt{2}}\left(\ket{00}\pm\ket{11}\right),\quad
    \ket{\psi_{\pm}}=\frac{1}{\sqrt{2}}\left(\ket{01}\pm\ket{10}\right).
\end{equation}
Now, to detect entanglement in Bell-basis measurement, a possible witness $\mathcal{W}_{BM}$ is given by $\mathcal{W}_{BM}=1/2\I-\proj{\phi_+}$. A way to implement this witness is to prepare the state $\ket{\phi_+}$ and act it on the measurement and if one satisfies the criterion \eqref{eq1}, then the measurement is confirmed to be entangled. However, entangled states is a costly resource and preparing particular entangled states is even more difficult. Thus, it is beneficial to express this witness in terms of local states which is much simpler to generate [see Fig. \ref{fig1}]. 

Any witness $\mathcal{W}$ can be decomposed into the form 
\begin{equation}
    \mathcal{W}=\sum_{ij\ldots} c_{ij\ldots}\proj{a_i}\otimes\proj{b_j}\otimes\ldots
\end{equation}
where $\proj{a_i},\proj{b_j},\ldots$ are positive semi-definite for any $i,j,\ldots$. Normalising and absorbing the factors into $c_{ij\ldots}$, we can safely conclude that they are valid density matrices. However, this might not be optimal in the sense that one might be required to measure a large number of correlations to conclude that the measurement is entangled. For instance, the witness $\mathcal{W}_{BM}$ decomposes as
\begin{equation}
    \mathcal{W}_{BM}=\frac{1}{4}\left(\I-\sum_{i=0}^2\sum_{j,j'=0,1}(-1)^{j+j'}\proj{\mu_{i,j}}\otimes\proj{\mu_{i,j'}}\right)
\end{equation}
where $\ket{\mu_{0,j}}=\ket{j},\ket{\mu_{1,j}}=1/\sqrt{2}(\ket{0}+(-1)^j\ket{1}),$ and $\ket{\mu_{2,j}}=1/\sqrt{2}(\ket{0}+i(-1)^j\ket{1})$. Thus, at least $12$ correlations must be considered to evaluate whether the BM is entangled. 

As it turns out, the number of correlations to observe in the particular case of BM can be significantly reduced. Consider the following witness 
\begin{eqnarray}
     \mathcal{W'}_{BM}=\frac{3}{2}\I-\sum_{i,j=0,1}\proj{\mu_{i,j}}\otimes\proj{\mu_{i,j}}.
\end{eqnarray}
It is simple to check that $\Tr(\mathcal{W'}_{BM}\sigma)\geq0$ for any separable state $\sigma$ and $\Tr(\mathcal{W'}_{BM}\proj{\phi_+})=-1/2$. Moreover one has to measure only $4$ correlations now. Let us remark here we do not have a prove that the witness $\mathcal{W'}_{BM}$ is optimal. However, based on our extensive search this seems to be the case. Thus, we conjecture that the above witness $\mathcal{W'}_{BM}$ is an optimal witness of the BM. Similarly, one can find different optimal witnesses tailored to specific quantum measurements that can be implemented using local quantum states.

{\it{One-sided device-independent witness---}} Let us now proceed towards relaxing the trust in the input states to show that a given measurement is entangled. For this purpose, we utilise the recently contrived swap-steering scenario, which is the minimal scenario to detect any form of network nonlocality without inputs \cite{Sarkar2024networkquantum}. The swap-steering scenario consists of two parties namely, Alice and Bob in two different labs. Both of them receive subsystems from $N$ independent sources $S_j$ that prepare the states $\rho_{A_jB_j}$ for $j=1,\ldots,N$. Here $A_j, B_j$ denote the $N$ different subsystems of Alice and Bob respectively. Bob performs a single measurement $\{E_b\}$ on the received subsystems where $b$ denotes his outcomes such that $E_b$ acts on $\bigotimes_N\mathbb{C}^d$. Alice is trusted in this context, meaning the measurements she performs on her subsystems are well-defined and known. Here, we consider the measurements of the trusted party is $\mathcal{A}_{i_1\ldots i_N}=\{\tau_{i_1}\otimes\ldots\tau_{i_N},\I-\tau_{i_1}\otimes\ldots\tau_{i_N}\}$ where $\tau_{i_j}\in \mathbb{C}^{d}$ are projectors such that $\{\tau_{i_j}\}$ span the Hilbert spaces $\mathbb{C}^{d}$ and thus are tomographically complete. Consequently, we have that $i_j=\{1,\ldots,d^2\}$ for $j=1,\ldots,N$. As we trust one of the sides in the experiment, in general such scenarios are referred to as one-sided device-independent (1SDI) (see Fig. \ref{fig2}). 

Alice and Bob repeat the experiment enough times to obtain the correlations $\vec{p}=\{p(a,b|i_1\ldots i_N)\}$ where $p(a,b|i_1\ldots i_N)$ denotes the probability of obtaining outcome $a,b$ with Alice and Bob respectively given Alice's input $i_1\ldots i_N$. These probabilities can be computed in quantum theory as
\begin{equation}
p(a,b|i_1\ldots i_N)=\Tr\left[(M_{a|i_1\ldots i_N}\otimes E_b)\bigotimes_N\rho_{A_iB_i}\right]
\end{equation}
where $M_{a|x}$ denote the measurement elements of Alice corresponding to input $x$. 
It is important to recall that Alice and Bob can not communicate with each other during the experiment. 

If the correlations $\vec{p}$ admit a separable outcome-independent hidden state (SOHS) model \cite{Sarkar2024networkquantum, sarkar2024witnessingnetworksteerabilitybipartite}, then $p(a,b|i_1\ldots i_N)$ is given by
\begin{equation}\label{sohs}
\sum_{\lambda_1,\ldots,\lambda_N}\prod_{i=1}^Np(\lambda_i)\Tr{[M_{a|i_1\ldots i_N}\rho_{\lambda_1}\otimes\ldots\rho_{\lambda_N}]} p(b|\lambda_1,\ldots,\lambda_N).
\end{equation}
This assumption can be further weakened by considering classically correlated sources \cite{Sarkar2024networkquantum}, however, for simplicity, we consider here that all sources are independent.

\begin{figure}[t]
\begin{center}
\includegraphics[width=.8\linewidth]{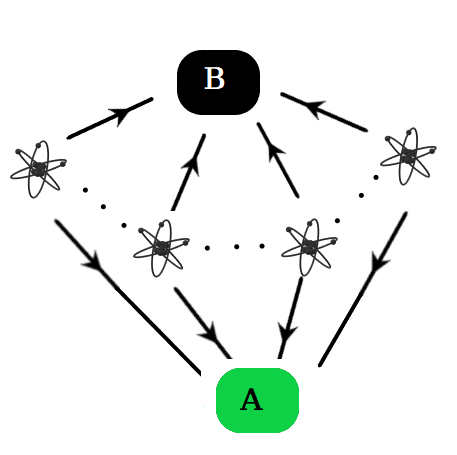}
    \caption{Swap-steering scenario. Alice and Bob are spatially separated and each of them receives $N$ subsystems from the untrusted sources. On the received subsystem Bob performs a single measurement while Alice being trusted performs a tomography on her subsystem. From the experiment, they obtain the correlations $\{p(a,b|i_1\ldots i_N)\}$.}
    \label{fig2}
    \end{center}
\end{figure}

Consider now an entangled measurement $\{E_b\}$ whose entanglement witness is given by $\mathcal{W}$ such that it satisfies the criterion \eqref{eq1}. The operator $\mathcal{W}$ can be expressed using the tomographically complete set of operators stated above as $\mathcal{W}=-\sum_{i_1,\ldots i_N=1}^{d^2}\beta_{i_1,\ldots i_N}\tau_{i_1}\otimes\ldots\tau_{i_N}$. Utilising this entanglement witness, let us now propose the following swap-steering inequalities 
\begin{eqnarray}\label{Wituniv}
    \mathcal{S}_{\{E_b\}}=\max_{b}\sum_{i_1,\ldots i_N}\beta_{i_1,\ldots i_N}p(0,b|i_1\ldots i_N)
\end{eqnarray}
where $p(0,b|i_1\ldots i_N)$ is the probability of obtaining outcome $0$ by trusted Alice given the input $i_1\ldots i_N$ and outcome $b$ by Bob.
Let us now find its SOHS bound.
\begin{fakt}
    Consider the swap-steering scenario in Fig. \ref{fig2} and the functional $\mathcal{S}_{\{E_b\}}$ \eqref{Wituniv}. The maximal value attainable of $\mathcal{S}_{\{E_b\}}$ using an SOHS model is $\beta_{SOHS}=0$. 
\end{fakt}

The proof of the above fact is in the Appendix.
Let us now consider that all the sources produce the maximally entangled state of dimension $d$, that is,  $\ket{\psi_{A_i,B_i}}=\ket{\phi^+_d}=1/\sqrt{d}(\sum_{i}\ket{ii})$ with Bob performing the measurement $\{E_b\}$. Let us say that Bob obtains the outcome $b$. Due to entanglement swapping, the post-measurement state at Alice is $E_b$ with a probability $\Tr(E_b)/d^{N}$ and thus one is left with the expression
\begin{eqnarray}
    \mathcal{S}_{\{E_b\}}&=&\max_{b}\frac{\Tr(E_b)}{d^N}\sum_{i_1,\ldots i_N}\beta_{i_1,\ldots i_N}\Tr(\tau_{i_1}\otimes\ldots\tau_{i_N} E_b)\nonumber\\&=&\max_{b}\frac{\Tr(E_b)}{d^N}\Tr(-\mathcal{W}E_b)>0.
\end{eqnarray} 
Consequently, we can conclude that every entangled measurement generates swap-steerable correlations and thus entanglement in any quantum measurement can be witnessed in a 1SDI way. For simplicity, we considered that the dimension of the local subsystems on which $\{E_b\}$ acts is the same. However, all the above arguments can straightaway be generalised to the case when they are different by considering that the source generates the maximally entangled state of the corresponding dimension of the local subsystem and Alice performs tomography on it. 

{\it{Device-independent witness---}} Let us now present a device-independent way to witness the entanglement in projective measurements. For this purpose, we consider the star network scenario \cite{Tavakoli_2014, sarkar2024universal} which is a generalisation of the bilocality scenario introduced in \cite{pironio1}. The star network is composed of an external group of $N$ Alices, $A_i$ $(i=1,\ldots, N)$, and a central party Bob, B. There are $N$ independent sources which distribute bipartite quantum states $\rho_{A_iB_i}$. Here, the subsystems $A_i$ go to the external parties on which they could perform local measurements, while all the subsystems $B_i$ go to Bob on which he can perform a joint measurement. The inputs and outputs of all the external parties are denoted as $x_i,a_i$ respectively with Bob performing a single measurement with outputs denoted as $b$. All the $N+1$ parties are spatially separated and can not communicate with each other (see Fig. \ref{fig3}).  

The experiment is repeated enough number of times to obtain the correlations $\vec{p}=\{p(\mathbf{a},b|\mathbf{x})\}$ where each $p(\mathbf{a},b|\mathbf{x})$ is the probability of obtaining outcomes $\mathbf{a}\equiv a_1\ldots a_N$ by all the Alices, given inputs $\mathbf{x}\equiv x_1,\ldots,x_N$, and $b$ by Bob. From Born's rule, we can write
\begin{equation}\label{probs}
p(\mathbf{a},b|\mathbf{x})=\Tr\left[\left(\bigotimes_{i=1}^{N} M_{a_i|x_i} \otimes E_b\right)\bigotimes_N\rho_{A_iB_i}\right].
\end{equation}
Here we consider $\{E_b\}$ to be rank-one projective measurement. If $p(\mathbf{a},b|\mathbf{x})$ admits a local description, then it can be expressed as \cite{pironio1}
\begin{equation}\label{bilocal}
p(\mathbf{a},b|\mathbf{x})=\sum_{\lambda_1,\ldots,\lambda_N}\prod_{i=1}^Np(\lambda_i)p(a_i|x_i,\lambda_i)p(b|\lambda_1,\ldots,\lambda_N).
\end{equation}
A weaker assumption on this network has also been studied in \cite{Sarkar_2024}.

Consider now the standard multipartite Bell scenario in which a single source distributes states to all the external parties $A_i$. On their respective subsystems, they perform measurements with inputs $x_i$ and outputs $a_i$. They gather enough statistics to obtain $p(\mathbf{a}|\mathbf{x})$ using which they can evaluate a Bell functional $\mathcal{B}=\sum_{\mathbf{a},\mathbf{x}}c_{\mathbf{a},\mathbf{x}}p(\mathbf{a}|\mathbf{x})$ with $c_{\mathbf{a},\mathbf{x}}$ being real. The Bell inequality is thus given by $\mathcal{B}\leq\beta_{LHV}$ where $\beta_{LHV}$ is the maximal value attainable using the local hidden variable (LHV) model, that is, $p(\mathbf{a}|\mathbf{x})=\sum_{\lambda}\prod_{i=1}^Np(a_i|x_i,\lambda)p(\lambda)$. Now, we show that using this Bell inequality $\mathcal{B}$, we can construct a witness in the star network scenario to detect the entanglement in the quantum measurement $\{E_b\}$.

\begin{figure}[t]
\begin{center}
\includegraphics[width=.8\linewidth]{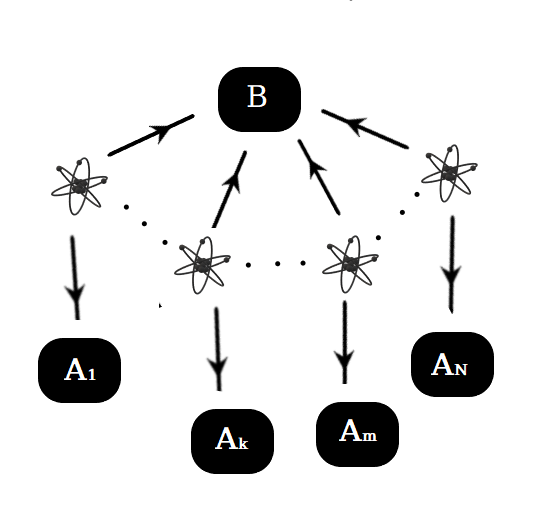}
    \caption{Star network scenario. There are $N$ external Alice's $A_i$ and Bob who are spatially separated. $A_i's$ receive one subsystem from their respective subsystem while Bob receives $N$ subsystems from all the sources. Bob performs a single joint measurement on his received subsystems while $A_i$ performs local measurements. Once the experiment is complete, they construct the joint probability distribution $\{p(\mathbf{a},b|\mathbf{x})\}$.}
    \label{fig3}
    \end{center}
\end{figure}

\begin{fakt} Consider a Bell inequality  $\mathcal{B}\leq\beta_{LHV}$ where $\mathcal{B}=\sum_{\mathbf{a},\mathbf{x}}c_{\mathbf{a},\mathbf{x}}p(\mathbf{a}|\mathbf{x})$ in the standard multipartite Bell scenario. 
Consider then the star network scenario (see Fig. \ref{fig3}) and the functional 
\begin{eqnarray}\label{DIentwit}
\mathbb{E}=\max_b\left[\sum_{\mathbf{a},\mathbf{x}}c_{\mathbf{a},\mathbf{x}}p(\mathbf{a},b|\mathbf{x})-\beta_{LHV}p(b)\right].
\end{eqnarray}
For correlations $\vec{p}$ admitting the local distribution \eqref{bilocal}, $\mathbb{E}\leq0$.
\end{fakt}
The proof of the above fact is stated in the Appendix. Consider now that the sources produce the state $\ket{\psi_{A_iB_i}}=\ket{\phi^+_d}$ and Bob performs the measurement $\{E_b\}$ to obtain the output $b$, then as discussed above the post-measurement state with the external parties is $E_b$ with a probability $1/d^N$. If $\mathcal{B}$ is violated with the state $E_b$, then $\mathbb{E}$ \eqref{DIentwit} will be violated using the measurement element $E_b$ and thus $\mathbb{E}>0$. 

It is well-known from Gisin's theorem \cite{GISIN1991201} that every pure bipartite entangled state violates the Clauser-Horne-Shimony-Holt (CHSH) \cite{CHSH} inequality. Consequently, every rank-one projective bipartite entangled measurement will violate the CHSH inequality modified as \eqref{DIentwit} in the star network with $N=2$ also known as the bilocality scenario. Again following \cite{POPESCU1992293, Gachechiladze_2017} similar conclusion can be arrived for any rank-one projective entangled measurement. Consequently, any rank-one projective entangled measurement will generate network non-local correlations and can be detected in a DI manner. Let us remark here that one can also detect higher rank projective entangled measurements or entangled POVM's by using the above procedure if a measurement element, which will now correspond to a mixed state violates some Bell inequality.

{\it{Discussions---}} In this work, we found that similar to detecting entangled quantum states, entanglement in any quantum measurement can also be detected without trusting them. In particular, we proposed an entanglement witness of the Bell-basis measurement that only requires to produce $\{\proj{0},\proj{1},\proj{+},\proj{-}\}$ which is simple to implement. For instance, in a polarisation-based setup, one has to simply rotate the polarisation of a single state $\ket{0}$ to prepare all the other states. We then extended this idea to witness entanglement in quantum measurements in a 1SDI way. For this purpose, we utilised the swap-steering scenario which is the minimal scenario to witness any form of network nonlocality without inputs. Finally, using the star network and Gisin's theorem, we constructed witnesses to detect entanglement in any rank-one projective quantum measurement in a DI manner. 

Several interesting problems follow-up from this work. The first one concerns finding optimal witnesses for several other entangled measurements. A challenging problem in this direction will be to witness more structures in the quantum measurement. For instance, is it possible to detect that every measurement element is entangled or not. Moreover, for quantum measurements that can act on more than two subsystems, it will be interesting to detect genuine entanglement in such measurements. Also, it will be interesting if other properties such as nonlocality without entanglement can be witnessed without trusting the measurement.

{\textit{Acknowledgements---}}
This project was funded within
the QuantERA II Programme (VERIqTAS project) that
has received funding from the European Union’s Horizon 2020 research and innovation programme under
Grant Agreement No 101017733.

\providecommand{\noopsort}[1]{}\providecommand{\singleletter}[1]{#1}%

\onecolumngrid
\appendix

\section{Proofs of the facts}

\setcounter{fakt}{1}
\begin{fakt}
    Consider the swap-steering scenario in Fig. \ref{fig2} and the functional $\mathcal{S}_{\{E_b\}}$ \eqref{Wituniv}. The maximal value attainable of $\mathcal{S}_{\{E_b\}}$ using an SOHS model is $\beta_{SOHS}=0$. 
\end{fakt}
\begin{proof}
    Let us recall that for correlations admitting a SOHS model, we have that
    \begin{equation}
        p(0,b|i_1\ldots i_N)= \sum_{\lambda_1,\ldots,\lambda_N}p(\lambda_1)\ldots p(\lambda_N)\Tr{[M_{0}\rho_{\lambda_1}\otimes\ldots\rho_{\lambda_N}]} p(b|\lambda_1,\ldots,\lambda_N)
    \end{equation}
for all $b$. Consequently, we have from \eqref{Wituniv} that
\begin{eqnarray}\label{A21}
    \mathcal{S}_{\{E_b\}}&=&\max_b\sum_{\lambda_1,\ldots,\lambda_N}p(\lambda_1)\ldots p(\lambda_N)\Gamma(\rho_{\lambda_1}\otimes\ldots\rho_{\lambda_N}) p(b|\lambda_1,\ldots,\lambda_N)\ \ 
\end{eqnarray}
where 
\begin{eqnarray}\label{gamma11}
    \Gamma(\rho_{\lambda_1}\otimes\rho_{\lambda_n})=\sum_{i_1,\ldots i_N}\beta_{i_1,\ldots i_N}p(0|i_1\ldots i_N,\rho_{\lambda_1}\otimes\ldots\rho_{\lambda_N}).
\end{eqnarray}
Expanding the above formula \eqref{gamma11}, by recalling Alice's measurements $\mathcal{A}_{i_1\ldots i_N}$, we obtain
\begin{eqnarray}
    \Gamma(\rho_{\lambda_1}\otimes\ldots\rho_{\lambda_N})&=& \sum_{i_1,\ldots i_N}\beta_{i_1,\ldots i_N}\Tr(\tau_{i_1}\otimes\ldots\tau_{i_N} \rho_{\lambda_1}\otimes\ldots\rho_{\lambda_N}) \nonumber\\&=&-\Tr(\mathcal{W}\rho_{\lambda_1}\otimes\ldots\rho_{\lambda_N})\leq 0.
\end{eqnarray}
Thus, we have from \eqref{A21} that for correlations admitting a SOHS model
\begin{eqnarray}
    \mathcal{S}_{\{E_b\}}\leq0.
\end{eqnarray}
This concludes the proof.
\end{proof}

\begin{fakt} Consider a Bell inequality  $\mathcal{B}\leq\beta_{LHV}$ where $\mathcal{B}=\sum_{\mathbf{a},\mathbf{x}}c_{\mathbf{a},\mathbf{x}}p(\mathbf{a}|\mathbf{x})$ in the standard multipartite Bell scenario. 
Consider then the star network scenario (see Fig. \ref{fig3}) and the functional 
\begin{eqnarray}
\mathbb{E}=\max_b\left[\sum_{\mathbf{a},\mathbf{x}}c_{\mathbf{a},\mathbf{x}}p(\mathbf{a},b|\mathbf{x})-\beta_{LHV}p(b)\right].
\end{eqnarray}
For correlations $\vec{p}$ admitting the local distribution \eqref{bilocal}, $\mathbb{E}\leq0$.
\end{fakt}
\begin{proof}
    For any correlation admitting a local distribution \eqref{bilocal}, $\mathbb{E}$ can be expressed as
    \begin{eqnarray}\label{A6}
\mathbb{E}=\max_b\left[\sum_{\mathbf{a},\mathbf{x}}c_{\mathbf{a},\mathbf{x}}\sum_{\lambda_1,\ldots,\lambda_N}\prod_{i=1}^Np(\lambda_i)p(a_i|x_i,\lambda_i)p(b|\lambda_1,\ldots,\lambda_N)-\beta_{LHV}p(b)\right].
    \end{eqnarray}
By representing $\prod_{i=1}^Np(\lambda_i)\equiv p(\lambda_1,\ldots,\lambda_N)$, we obtain using Bayes rule that 
\begin{equation}
    p(b|\lambda_1,\ldots,\lambda_N)p(\lambda_1,\ldots,\lambda_N)=p(b)p(\lambda_1,\ldots,\lambda_N|b).
\end{equation}
Consequently, we obtain from \eqref{A6} that
\begin{eqnarray}
\mathbb{E}=\max_b\left[\sum_{\mathbf{a},\mathbf{x}}c_{\mathbf{a},\mathbf{x}}\sum_{\lambda_1,\ldots,\lambda_N}\prod_{i=1}^Np(a_i|x_i,\lambda_i)p(\lambda_1,\ldots,\lambda_N|b)-\beta_{LHV}\right]p(b).
    \end{eqnarray}
As $\lambda_1,\ldots,\lambda_N$ depending on $b$ represents another set of variables $\lambda'$ which in general might not be independent, we obtain that
\begin{eqnarray}
\mathbb{E}=\max_b\left[\sum_{\mathbf{a},\mathbf{x}}c_{\mathbf{a},\mathbf{x}}\sum_{\lambda'}\prod_{i=1}^Np(a_i|x_i,\lambda')p(\lambda')-\beta_{LHV}\right]p(b).
    \end{eqnarray}
    As the first term in the above formula for correlations admitting the LHV model in the standard multipartite Bell scenario, we obtain that $\mathbb{E}\leq0$. This completes the proof.
\end{proof}

\end{document}